\documentclass[11pt]{article}

\usepackage{amssymb, amsthm}
\usepackage{latexsym}
\usepackage{amsmath}
\usepackage[ansinew]{inputenc}
\usepackage{hyperref}

\numberwithin{equation}{section}

\newtheorem{theo}{Theorem}[section]

\newtheorem{lem}[theo]{Lemma}
\theoremstyle{definition}
\newtheorem{defi}[theo]{Definition}
\newtheorem{bsp}[theo]{Example}

\begin{document}
\begin{titlepage}

\date{\today}
\begin{center}
{\bf\Large On the relation between geometric\\ and deformation
quantization}

\vskip1.5cm

{\normalsize Christoph Nölle}

\vskip1cm

{\it Institut für Theoretische Physik und Astrophysik, Universität Kiel, \\
 Leibnizstraße 15, 24105 Kiel, Germany}

\end{center}

\vskip1.5cm

\begin{abstract}
  In this paper we investigate the possibility of constructing a
  complete quantization procedure consisting of geometric and deformation
  quantization. The latter assigns a noncommutative
  algebra to a symplectic manifold, by deforming the
  ordinary pointwise product of functions, whereas geometric
  quantization is a prescription for the construction of a Hilbert
  space and a few quantum operators, starting from a symplectic
  manifold. We determine under which conditions it is possible to define a representation
  of the deformed algebra on this Hilbert space, thereby to extend the small class of quantizable
  observables in geometric quantization to all smooth functions, as
  well as to give a natural representation of the algebra.
  In particular we look at the special cases of a cotangent bundle and a Kähler
  manifold.

\end{abstract}
\end{titlepage}

%\tableofcontents

\section{Introduction}

 There are two important concepts for quantization of arbitrary
 symplectic manifolds $(M,\omega)$, or phase spaces. Geometric quantization
 was developed in an attempt to unify several special quantization
 schemes that arose in applications (see \cite{woodhouse:quantisierung} and references therein), and like in ordinary quantum
 mechanics on a flat phase space, physical states are elements
 of a Hilbert space $\mathcal H$, explicitly represented as wave functions on a subspace of phase
 space. This subspace is not uniquely determined, but depends on the
 choice of a polarization, although for finite-dimensional spaces different polarizations are supposed to
 lead to equivalent theories, like position and momentum space
 representations. There is a problem with the observables though,
 because for a fixed polarization only a very limited class of them can be
 quantized.

\bigskip

 On the other hand there is deformation quantization \cite{bayenundvieleandere1,
 bayenundvieleandere2, dito-sternheimer_deformueberblick}, whose
 philosophy is to introduce a new product $\ast$, called a star
 product, on some function algebra over $M$, say $C^\infty(M)$.
 The condition for the star product to be a deformation of the
 ordinary product is expressed by demanding
  $$ f\ast g = fg + \mathcal O(\hbar)$$
 for all $f,g\in C^\infty(M)$, so that Planck's constant $\hbar$ plays the role
 of a deformation parameter. Further the condition
  $$ [f,g]_\ast:=f\ast g-g\ast f=i\hbar\{f,g\} +\mathcal
  O(\hbar^2)$$
 is imposed, where $\{\cdot,\cdot\}$ denotes the classical Poisson
 bracket, generalizing Dirac's quantum condition \cite{dirac:fundeqqm}.
 Kontsevich \cite{kontsevich_deformPoisson} has proved that any Poisson manifold admits a star
 product, but for the case of a symplectic manifold there exists a
 much simpler construction, due to Fedosov \cite{fedosov:deformationart,fedosov:deformationbook}. We will
 consider only the symplectic case here. It is notable that the
 star product is not uniquely determined by $(M,\omega)$ either, see the brief review of
 Fedosov's construction in section \ref{DefoQuant}.

  Giving the full
 construction of the algebra of quantum observables $\mathcal A$,
 deformation quantization has the drawback not to offer a physical
 interpretation of the states. As for $C^\ast$ algebras states can be defined as
 positive linear functionals on the algebra \cite{bordem_statesandGNS}, but
 they do not have a natural identification as square roots of volume elements on 'configuration
 space', as in geometric quantization.

 \bigskip

 Therefore it would be desirable to have a representation of the
 deformed algebra $\mathcal A$ on the Hilbert space $\mathcal H$,
 thus offering a physical interpretation of states in
 deformation quantization, and extending the algebra of observables
 in geometric quantization. As the latter already defines a
 quantization map $\rho'$ for a subset of $C^\infty(M)$, one might hope to find a unique
 extension $\rho$ of $\rho'$ to all of $\mathcal A$ by requiring it to be an
 algebra homomorphism. There is then an obvious compatibility
 condition to be fulfilled, viz. that the original
 $\rho'$ respects the star product:
  \begin{equation}\label{kompatibilitaetIntro}
         \rho'(f\ast g) = \rho'(f)\ast \rho'(g)
  \end{equation}
 for any $f,g\in C^\infty(M)$ with the property that $f,g$ and
 also $f\ast g$ are geometric quantizable. As neither geometric nor
 deformation quantization are uniquely determined by $\omega$, but
 depend on further structure, one should not expect equation
 (\ref{kompatibilitaetIntro}) to hold for any arbitrary choice of
 these structures, but rather to give a compatibility condition for
 them. The best possible result would then be to get a unique (class of)
 solution(s) to the compatibility condition, thus fixing the ambiguities
 in the separate quantization procedures.

\bigskip

 Usually deformation and geometric quantization are considered as
 competing theories, and deformation quantization is often referred
 to as the more promising candidate for a consistent theory of
 quantization, e.g. in \cite{dito-sternheimer_deformueberblick}.
 A point of view similar to ours is taken by Landsman in \cite{landsman_classQuant}.

\bigskip

 We will restrict our attention to two important special cases, a
 cotangent bundle $T^\ast Q$ with the canonical two-form, and a
 Kähler manifold. For cotangent bundles, equation (\ref{kompatibilitaetIntro}) is a direct
 consequence of a result in \cite{BordemannNW-StarprodKot}, whereas
 for Kähler manifolds with their most natural polarization and connection, a full proof
 remains out of reach. Therefore, we restrict ourselves to providing some evidence in favor
 of compatibility.

\section{Geometric quantization}
 This section contains a concise summary of the theory,
 as it is presented in the monograph
 \cite{woodhouse:quantisierung}. Starting with a symplectic manifold
 $(M,\omega)$, our aim is the construction of a Hilbert space $\mathcal H$ and
 a quantization map $f\mapsto \hat f$, assigning an operator on
 $\mathcal H$ to a smooth function $f$ on $M$. The following conditions
 are to be satisfied:
    \begin{description}
    \item[\qquad(i)] linearity: $\widehat{\lambda f+g}=\lambda \hat
    f+\hat g$, $\lambda \in \mathbb C$.
    \item[\qquad(ii)] $f$ constant implies $\hat f=f\mathbf 1$,
    \item[\qquad(iii)] the correspondence principle: $[\hat f,\hat g] =
    i\hbar\widehat{\{f,g\}}$.
  \end{description}
 There is a simple construction obeying these conditions, consisting
 of a Hermitian line bundle $B\rightarrow M$ with metric connection of curvature $-\frac
 i\hbar \omega$, the Hilbert space $\mathcal H=L^2(M,B)$ and
  \begin{equation}
         \hat f=-i\hbar \nabla_{X_f} +f\mathbf 1, \qquad f\in
         C^\infty(M).
  \end{equation}
 The Hamiltonian vector field $X_f=df^\sharp$ is defined by means of the
 musical isomorphism
  $$ \sharp :T^\ast M\rightarrow TM,\
  \eta(\phi^\sharp)=\omega^{-1}(\eta,\phi)=\omega^{ab}\eta_a\phi_b,$$
 and $\nabla$ is the connection on $B$. In a trivialization it is
 expressed as $\nabla=d-\frac i\hbar \theta$, where $\theta$ is a
 symplectic potential, obeying $d\theta=\omega$. $B$ is called a pre-quantum bundle.\\
 It turns out that the Hilbert space $L^2(M,B)$ is too large, as it
 contains functions of both $q$ and $p$ in the flat case, and gives rise to a reducible
 representation of the Heisenberg algebra. A method
 to reduce degrees of freedoms is needed, and this is provided by a
 polarization:

\begin{defi}[polarization]
    A polarization of a symplectic manifold $(M,\omega)$ is an
   integrable subbundle of $T_\mathbb CM$ (also called an involutive
   distribution), with the properties that
  \begin{enumerate}
   \item Every subspace $P_m\subset {T_m}_\mathbb C M$ is a Lagrange
    subspace, i.e. the symplectic complement
     $$ P_m^\perp =\{v\in {T_m}_\mathbb C M\ \big|\ \omega(v,p)=0\ \forall p\in
     P_m \}$$
    satisfies $P^\perp=P$. In particular this implies dim$_\mathbb C
    P$=dim$_\mathbb R M/2$.
   \item The involutive distribution $D=(P\cup \overline P)\cap TM$
    has constant real dimension $d$, which is called the real index of
    $P$.
   \item The distribution $E=(P+ \overline P)\cap TM$ is
   involutive.
 \end{enumerate}
 A function $f\in C^\infty(M)$ is called polarized if $\overline
 X(f)=0$ holds for any $X\in\Gamma(P)$. A symplectic potential
 $\theta$ is called adapted to $P$, if $\theta(\overline X)=0$ for
 $X\in \Gamma(P)$.
\end{defi}

 It is non-standard to include the requirement that $E$ be
 involutive into the definition, usually one speaks of a strongly
 integrable polarization \cite{woodhouse:quantisierung}, or a Nirenberg integrable subbundle
 \cite{gotay_NonexistencePol} in case it is satisfied. However, this condition is essential for our
 purposes, and implies the existence of an adapted symplectic
 potential. If $d=n=$ dim $M/2$ the polarization is called real, and
 is the complexification of a real subbundle of $TM$. At the
 other extreme there are the totally complex polarizations with $d=0$. They
 have the property $E=TM$, and any vector $X\in T_mM$ can be
 uniquely decomposed as $X=Z+\overline Z$, where $Z\in P_m$. Then
  $$ J:TM\rightarrow TM,\ Z+\overline Z\mapsto iZ-i\overline Z$$
 defines a complex structure on $TM$, which is compatible with
 $\omega$, i.e. is a symplectic transformation. $J$ in turn induces
 a semi-Riemannian metric
  $$ g(X,Y)=\omega(X,JY)$$
 on $TM$, which, if positive definite, defines a Kähler structure on
 $M$. On the other hand, every Kähler manifold has a canonical
 polarization: $P_m=\{X-iJX \ |\ X\in T_mM\}$. The main example
 of a real polarization is given by the vertical foliation of a
 cotangent bundle, where $P_m$ is the vertical, or fibre-parallel
 subspace of $T_m(T^\ast Q)$.

\begin{defi}[adapted coordinates]
 Let $P\subset T_\mathbb CM$ be a polarization, then in a neighborhood of any point of $M$ there are
 coordinates $q^i,p_j,z^\alpha$, where $q$, $p$ are real and $z$
 complex, such that $P$ is spanned by the $\partial_{p_j}$ and
 $\partial_{z^\alpha}$. Further, a standard form for $\omega$ can be
 obtained (\cite{woodhouse:quantisierung}, p. 97), which for a real polarization
 takes the form $\omega= dp_i\wedge dq^i$ (then $q,p$ are called Darboux coordinates), and for a Kähler
 polarization
  $$ \omega=i\frac{\partial^2 K}{\partial z^\alpha\partial \overline
  z^\beta} dz^\alpha \wedge d\overline z^\beta.$$
 The real function $K$ is called a Kähler potential. These coordinates
 are said to be adapted to $P$.
\end{defi}
 In the real case a polarized function depends only on the $q$s,
 whereas a polarized function on a Kähler manifold is holomorphic.

\begin{defi}[canonical bundle]
 Let $P$ be a polarization of $(M,\omega)$, dim $M=2n$, then the
 canonical bundle $K_P\rightarrow M$ is a complex line bundle over
 $M$, with fibre
  \begin{equation}
    (K_P)_m=\{\alpha\in \Lambda^n {T^\ast_m}_\mathbb C M\ \big|\
    \iota_{\overline X}\alpha =0 \ \forall X\in \Gamma(P)\}.
  \end{equation}
 $\iota_{\overline X}$ denotes contraction with the vector $\overline X$.
\end{defi}
 We also consider $K$ as a bundle over the space $LM$ of all nonnegative
 polarizations of $M$ (see \cite{woodhouse:quantisierung}), with fibre $K_P$ at $P\in
 LM$. A square root of $K\rightarrow LM$ is a line bundle
 $\delta\rightarrow LM$ satisfying
  $$ \delta^2:=\delta\otimes \delta=K.$$
 According to Kostant such a bundle exists iff $M$ admits a
 metaplectic structure \cite{kostant_symplecticspinors}.
 In the following we will assume that $M$ carries a metaplectic
 structure, and hence a square root $\delta$ of $K$. To a given
 (nonnegative) polarization $P$ we assign the so called half-form bundle
 $\delta_P\rightarrow M$, with fibre $(\delta_P)_m$ at $m\in M$.
 Here $LM$ is considered as a bundle over $M$, whose fibre at $m$
 consists of all polarizations of ${T_m}_\mathbb CM$.\\
 The Lie derivative of a vector field $X$ preserves sections of
 $K_P$ iff the flow of $X$ preserves the polarization $P$, which is
 equivalent to $[X,Y]\in \Gamma(P)$ for any $Y\in \Gamma(P)$. In
 this case the Lie derivative $\mathcal L_X$ is transferred from $K_P$ to $\delta_P$ by
  \begin{equation}
    2(\mathcal L_X\nu)\nu=\mathcal L_X\nu^2 ,\qquad \forall
    \nu\in\Gamma(\delta_P).
  \end{equation}
 On $K_P$ the partial connection $\nabla$ is defined by
  \begin{equation}
    \nabla_{\overline X}\beta=\iota_{\overline X}d\beta,\qquad
    \forall X\in\Gamma(P),\ \beta\in\Gamma(K_P),
  \end{equation}
 which is also transferred to $\delta_P$ by the requirement
 $2(\nabla_{\overline X}\nu)\nu=\nabla_{\overline X}\nu^2$. There is
 a natural sesquilinear product $(\cdot,\cdot)$ on $\delta_P$,
 taking values in the densities on $M/D$, for the construction see
 \cite{woodhouse:quantisierung}, p.230 (assuming that $M/D$ is a Hausdorff manifold).
 We consider the bundle $B_P=B\otimes \delta_P$, equipped with the
 partial connection induced by the connection on $B$ and the partial
 connection on $\delta_P$, and take as state space the
 set of 'polarized wave functions'
  \begin{equation}\label{geomHilbertSpace}
        \Gamma_P(B_P)=\{s\in \Gamma(B_P)\ \big|\
        \nabla_{\overline X}s=0 \ \forall X\in \Gamma(P)\}.
  \end{equation}
 The prescription
  \begin{equation}
        \langle \psi\otimes \mu, \phi\otimes \nu\rangle=\int
        _{M/D}(\psi,\phi)(\mu,\nu)
  \end{equation}
 defines an inner product on $\Gamma_P(B_P)$. Our Hilbert space
 $\mathcal H_P$ is the set of square integrable elements of
 $\Gamma_P(B_P)$. The quantization of an observable $f$ becomes
  \begin{equation}\label{geomOperatorDefi}
    \tilde f(\psi\otimes\nu) =(\hat f\psi)\otimes \nu -i\hbar \psi\otimes \mathcal L_{X_f}\nu,
  \end{equation}
 in case $X_f$ preserves the polarization. For $P$ real, these are functions of the form
  \begin{equation}\label{ObsAffLinReell}
        f(q,p)=a^i(q)p_i+b(q),
  \end{equation}
 in adapted coordinates, whereas on Kähler manifolds with holomorphic polarization the condition
 implies
  $$ f(z,\overline z)=u^a(z)\frac{\partial K}{\partial z^a}+ v(z).$$
\begin{bsp}
 On a cotangent bundle $M=T^\ast Q$ over a (semi-)Riemannian manifold $(Q,g)$, with
 coordinates $q$ on $Q$ and $(q,p=dq)$ on $M$, an adapted potential
 is given by $\theta=p_idq^i$. In this gauge, wave functions can be expressed
 as $\psi\otimes \sqrt{\mu}$, where $\psi=\psi(q)$ is polarized, and
 $\mu=\sqrt{|\det g(q)|}d^n q$.
 Then we have for $f$ as in (\ref{ObsAffLinReell}):
  \begin{equation}
   \tilde f(\psi\otimes \sqrt\mu)=\frac \hbar i\Big(a^j\partial_j
   \psi + \big(b+\textstyle{\frac
   12}\text{div}(a)\big)\psi\Big)\otimes \sqrt\mu,
  \end{equation}
 where div$(a)$ is defined by $\mathcal
 L_a\mu=d(\iota_a\mu)=:$ div$(a)\mu$. In particular this gives
 (neglecting $\sqrt \mu$)
  \begin{equation}\label{kanOpsGeomQuant}
   \tilde q^a\psi(q)=q^a\psi(q), \quad \text{and} \quad \tilde
   p_a\psi (q)=-i\hbar \big(\partial_{a}+\textstyle{\frac
   14}g^{bc}\partial_ag_{bc}\big) \psi(q).
  \end{equation}
 A monomial $q^{i_1}\dots q^{i_k}p_j$ is mapped to
 the symmetrized product of the operators $\tilde
 q^{i_1},\dots,\tilde q_{i_k},\tilde p_j$, reproducing the Weyl ordering convention.
\end{bsp}
\begin{bsp}
 On a Kähler manifold with adapted potential $\theta=-i\frac {\partial
 K}{\partial z^a}dz^a$ and half-form $\nu=\sqrt{d^nz}$, wave functions are holomorphic, and the
 metric on the pre-quantum bundle is $\langle
 \psi,\phi\rangle=e^{- K/\hbar}\overline \psi\phi$. One gets $\tilde z^a
 \psi(z)=z^a\psi(z)$, and the function $f(z,\overline
 z)=\frac{\partial K}{\partial z^a}$ is quantized as:
  \begin{equation}
   \tilde f\psi(z) =\hbar \partial_{a}\psi(z).
  \end{equation}
 Again, an observable linear in $\partial_{z^a}K$ gives rise to a
 symmetrized operator.
\end{bsp}

\section{Deformation quantization}\label{DefoQuant}
 The purpose of deformation quantization is to construct a star
 product on $C^\infty(M)$, i.e. a deformation of the algebra of
 classical observables on a symplectic manifold. Our main references
 are \cite{fedosov:deformationart,fedosov:deformationbook}. The following conditions are supposed to
 hold, where $\hbar$ is considered as a formal deformation
 parameter \cite{bayenundvieleandere1, dito-sternheimer_deformueberblick}:
  \begin{enumerate}
    \item the coefficients $c_k$ of the product of
    $f=\sum_k\hbar^kf_k$ and $g=\sum_k\hbar^k g_k$:
        $$ c = f \ast g = \sum_{k=0}^\infty
        \hbar^kc_k(f,g)$$
      are bi-differential operators (of finite order).
    \item $c_0(x)=f_0(x)g_0(x)$, i.e. $\ast$ is a deformation
      of the ordinary pointwise product on $C^\infty(M)$.
    \item the correspondence principle
      $$ [f,g]_\ast =f\ast g-g\ast f = i\hbar \{f_0,g_0\} +
      O(\hbar^2), $$
     holds, where $\{f,g\}=\omega^{ab}\partial_a f\partial_bg$ denotes the poisson bracket defined by
     $\omega$.
  \end{enumerate}
 Fedosovs construction of $\ast$ makes use of a bijection from
 $C^\infty(M)$ to the set of flat, smooth sections of a vector bundle, we
 give a brief review of the method. \\
 First we need a symplectic (torsion-free) connection on the tangent bundle $TM$,
 i.e. a covariant derivative $\nabla$ respecting $\omega$:
  $$ d\big(\omega(X,Y)\big)=\omega(\nabla X,Y)+\omega(X,\nabla Y),$$
 or $\nabla \omega=0$ for the induced connection on $T^\ast M\otimes
 T^\ast M$. Such a connection always exists, but contrary to the
 Riemannian (symmetric) case it is not unique. We will later comment
 on the correct choice of $\nabla$. In local Darboux coordinates it takes the form
  $ \nabla=d+\Gamma,$
 where the connection form $\Gamma$ is a 1-form with values in the symplectic Lie
 algebra, i.e. $\Gamma\in \Omega^1\big(U;\mathfrak{sp}(2n)\big)$.
 As usual, its components are defined by $\Gamma^k_{ij}e_k=
 \Gamma(e_i)\cdot e_j$, where $e_i$ ($i=1,\dots,2n$) is the local Darboux
 basis of $TM$, $e_i=\partial_{q^i}$ and
 $e_{n+i}=\partial_{p^i}$ for $i=1,\dots,n$. Further we set
 $\Gamma_{ijk}:=\omega_{il}\Gamma^l_{jk}$, which is totally
 symmetric in its three indices due to $\Gamma$ being torsion-free
 and symplectic. The components of the curvature tensor are
    \begin{equation*}
        {R^i}_{jkl} =
        \partial_k\Gamma^i_{lj}-\partial_l\Gamma^i_{kj} +
        \Gamma^i_{km} \Gamma ^m_{lj} - \Gamma^i_{lm} \Gamma ^m_{kj},
  \end{equation*}
 and $R_{ijkl}=\omega_{im}{R^m}_{jkl}$ is symmetric in the first two
 (Lie algebra) indices, and antisymmetric in the last two
 (differential form) indices. $\nabla$ also induces a connection on
 $T^\ast M$, with connection form $-\Gamma^T$, and on the whole
 tensor algebra through the Leibniz formula. Particularly we are interested in the symmetric algebra
 $\mathcal W:=$Sym$_\mathbb C(T^\ast M)$; instead of the common $dx^i$ we use the symbols
 $y^i$ to denote the local basis of $T^\ast M$, and suppress the
 symmetric tensor product sign: $y^iy^j=y^jy^i:=\frac 12(y^i\otimes y^j+y^j\otimes y^i)$.
 In some cases we need the product bundle $\mathcal W\otimes \Lambda:=\mathcal W\otimes
 \Lambda(T^\ast M)$ and $\Omega(\mathcal W)=\Gamma(\mathcal W\otimes \Lambda)$,
 where the basis of $T^\ast M$ considered as a subset of
 $\Lambda (T^\ast M)$ is denoted by $dx^i$.
 Now we introduce a new product $\circ$ on $W_m=$Sym$_\mathbb C(T_m^\ast M)$, through the relations
  \begin{equation}\label{SympCliffRels}
     [v,w]_\circ=v\circ w-w\circ v=i\hbar \omega^{-1}(v,w),\quad
      v\circ w+w\circ v = 2vw,\quad \forall v,w\in T^\ast_m M.
  \end{equation}
 Then we have $[y^i,y^j]_\circ=i\hbar\omega^{ij}$, so that the $y^i$
 can be identified as local position and momentum operators.
 $(W_m,\circ)$ is called the Weyl algebra associated to
 $(T^\ast_m M,\omega_m^{-1})$, it is the symplectic analog of a Clifford algebra. On
 $\Omega(\mathcal W)$ we write $\circ $ for $\circ \otimes \wedge$ and use the graded commutator
  $$[\xi,\eta] = \xi\circ \eta - (-1)^{pq}\eta\circ \xi$$
 for $\xi\in \Omega^q(\mathcal W)$ and
 $\eta\in \Omega^p(\mathcal W)$.
 The connection is extended to $\Omega(\mathcal W)$ by
    \begin{align*}
    \nabla (\xi\circ \eta) &=\nabla \xi\circ\eta + (-1)^q\xi
    \circ\nabla\eta, \qquad \xi \in\Gamma(\mathcal W\otimes\Lambda^q) \\
    \nabla (\phi\wedge \eta ) &=d\phi \wedge \eta + (-1)^q \phi\wedge
     \nabla \eta ,\qquad \phi\in \Omega^q(M).
  \end{align*}
 Explicitly we have
  $$ \nabla y^{i_1}\dots y^{i_k} = -\sum_j \Gamma^{i_j}_{ab}
  y^{i_1}\dots \breve y^{i_j} \dots y^{i_k}y^a dx^b,$$
 where $\breve y^{i_j}$ means omitting the element, and which can be written in the form
  \begin{equation}
    \nabla  = d+ [dU(\Gamma),\cdot] := d-\frac i{2\hbar}
    \Gamma_{ijk}[y^iy^j,\cdot]dx^k.
  \end{equation}
 Here $dU$ is the isomorphism between the symplectic and
 metaplectic Lie algebras \cite{folland:harmonic_analysis_in_ps}:
  $$ dU: \mathfrak{sp}(2n)\rightarrow \mathfrak{mp}(2n),\ A\mapsto
  -\frac i{2\hbar}\omega_{ij}{A^j}_k y^iy^k.$$
 Now we introduce two further operators on $\Omega(\mathcal W)$:
  \begin{equation}
    \delta = dx^k\wedge \frac{\partial}{\partial y^k},\qquad
    \delta^\ast= y^k \iota\Big(\frac{\partial}{\partial x^k}\Big),
  \end{equation}
 where $y^k$ denotes (commutative) multiplication with $y^k$, and
 the contraction $\iota\Big(\frac{\partial}{\partial x^k}\Big)$ acts
 only on the form part. In brief, $\delta$ replaces one of the $y^i$ by $dx^i$,
 whereas $\delta^\ast$ replaces $dx^i$ by $y^i$. One easily checks
\begin{lem}\
 \begin{enumerate}
   \item $\delta^2=(\delta^\ast)^2=0$
   \item Applied to $ y^{i_1}\dots y^{i_l} dx^{j_1}\wedge\dots \wedge dx^{j_p}
   $ the following identity holds:
     \begin{equation}
        \delta \delta^\ast+ \delta^\ast \delta =( l + p)id.
     \end{equation}
 \end{enumerate}
\end{lem}
 We also define $\delta^{-1}$ by
 \begin{equation}
        \delta^{-1} = \frac 1{l+p}\delta^\ast
 \end{equation}
 for $l+p>0$, and $\delta^{-1}=0$ otherwise. If the projection of an
 element $\xi \in\Omega(\mathcal W)$ to its part in
 $C^\infty(M)\subset\Omega(\mathcal W)$ is denoted by $\xi_{00}$, then
 the following decomposition holds, analogously to the Hodge-de
 Rahm decomposition of forms:
  \begin{equation}\label{decompFed}
    \xi = \delta\delta^{-1} \xi + \delta^{-1}\delta \xi + \xi_{00}.
  \end{equation}
 From now on we consider elements of $\Omega(\mathcal W)$ as formal series
 in $\hbar$ (replace $\mathcal W=$Sym$_\mathbb C(T^\ast M)$ by $\mathcal
 W = $Sym$_\mathbb C(T^\ast M)[[\hbar]]$),
 and also allow for an infinite number of homogeneous elements in
 the $y^i$. Then elements of $\Omega(\mathcal W)$ have the expansion
 $$    \xi = \sum_{k,l,m\geq 0} \hbar^k \xi_{k; i_1,\dots,i_l,j_1,\dots,j_m}
   y^{i_1} \dots y^{i_l}dx^{j_1}\wedge \dots\wedge dx^{j_m}.$$
\begin{defi}
 The $\hbar$-degree of a homogeneous element
  \begin{equation}
    \sum_m \hbar^k \xi_{k; i_1,\dots,i_l,j_1,\dots,j_m}
   y^{i_1} \dots y^{i_l}dx^{j_1}\wedge \dots\wedge dx^{j_m}
  \end{equation}
 (no sum over $k,l$) is defined as $k+l/2$.
\end{defi}
 Due to the relations (\ref{SympCliffRels}) $\circ$ respects the
 gradation defined by the $\hbar$-degree. The subspaces of
 homogeneous elements of $\hbar$-degree $j$ are denoted by $\mathcal
 W_j$. Now $\mathcal W$ and $\Omega(\mathcal W)$
 carry two gradations, the other one being defined by the degree in
 Sym$_\mathbb C(T^\ast M)$, which is respected by the symmetric tensor
 product, but not by $\circ$. The corresponding projections are
  \begin{equation}\label{GradProjektionen}
        \pi_\hbar^j :\mathcal W \rightarrow \mathcal W_j,\qquad \text{und}\qquad
         \pi_\otimes^k: \mathcal  W \rightarrow \text{Sym}^k_\mathbb C(T^\ast
         M)[[\hbar]],
  \end{equation}
 we also adopt the convention to denote $\pi_\otimes^0(\xi)$ as
 $\xi_0$. It is important to note that $\delta$
 decreases the $\hbar$-degree by $1/2$, whereas $\delta^{\ast}$
 increases it.

 Now we come to the construction of a flat connection
 on $\mathcal W$, with curvature $\Omega=\frac
 i\hbar[\omega,\cdot]=0$. Making the ansatz
  \begin{equation}
    D=\nabla +\frac i\hbar[\gamma,\cdot]= d+ [dU(\Gamma)+
    \frac i\hbar\gamma,\cdot]
  \end{equation}
 with an as yet undetermined 1-form $\gamma\in \Omega^1(\mathcal
 W)$, we obtain for the curvature $\Omega=D^2=\frac i\hbar[\tilde \Omega,\cdot]$,
 with
  \begin{equation}
   \tilde \Omega=\frac \hbar idU(R)+ \nabla \gamma + \frac
  i\hbar\gamma^2.
  \end{equation}
 Here $R$ denotes the curvature 2-form of $\nabla$, $dU(R) =-\frac
 i{4\hbar} R_{ijkl}y^iy^jdx^k\wedge dx^l$, and $\gamma^2=\gamma\circ\gamma$. Of course, $\gamma$ is
 only determined up to addition of a scalar form. We require the normalization
 $\gamma_0=0.$ In the flat case
 $M=\mathbb R^{2n}$ with standard symplectic form $\omega=dp_i\wedge dq^i$ we can choose
 $\Gamma=0$, and $\gamma=\omega_{ab}y^bdx^a$ leads to
 $\tilde\Omega=\frac i\hbar \omega$. Therefore, in the general case we
 split $\gamma$ as
  \begin{equation}
   \gamma=\omega_{ab}y^bdx^a + r.
  \end{equation}
 Observing that $\delta$ can be written in the form $\delta \xi =
 -\frac i\hbar \omega_{ab}dx^a[y^b,\xi]$, we obtain for the curvature
    $$  \tilde\Omega= \omega + \frac \hbar i dU(R)- \delta r
          +\nabla r +\frac i\hbar r^2,$$
 so that $\tilde \Omega=\frac i\hbar \omega$ becomes equivalent to
  \begin{equation}\label{abelscherZshgBedg}
    \delta r = \hat R +\nabla r + \frac
     i\hbar r^2,
  \end{equation}
 where $\hat R=\frac \hbar idU(R)=-\frac 14R_{ijkl}y^iy^jdx^k\wedge
 dx^l$ has been introduced.

 \begin{theo}[Fedosov]
  Under the condition $\delta^{-1}r=0$, eq.
  (\ref{abelscherZshgBedg}) has exactly one solution $r$.
 \end{theo}
\begin{proof}[Sketch of proof]
 For a 1-form we have $r_{00}=0$, together with the condition
 $\delta^{-1}r=0$ this implies that the decomposition (\ref{decompFed}) takes
 the form $r=\delta^{-1}\delta r$. Applying $\delta^{-1}$ to
 (\ref{abelscherZshgBedg}) leads to
   \begin{equation}\label{abelscherZshgBedg2}
     r = \delta^{-1} \hat R +\delta^{-1} \Big(\nabla r +
     \frac i\hbar r^2\Big).
   \end{equation}
 As $\nabla$ preserves the $\hbar$-filtration, whereas $\delta^{-1}$
 increases the degree by 1/2, it follows by iteration that
 (\ref{abelscherZshgBedg2}) has exactly one solution (as a formal
 power series in $\hbar$ and $y^i$). The iteration steps are
  \begin{equation}
    r^{(3)}=\delta^{-1}\hat R,\quad r^{(n+1)}=\delta^{-1}\hat R
    +\delta^{-1}\big(\nabla r^{(n)}+\frac i\hbar (r^{(n)})^2\big),
  \end{equation}
 and $r=\lim_{n\rightarrow \infty} r^{(n)}$. The condition
 $\delta^{-1}r=0$ is fulfilled due to $(\delta^{-1})^2=0$. For the
 proof that $r$ indeed solves (\ref{abelscherZshgBedg}) we refer to
 Fedosovs texts \cite{fedosov:deformationart,fedosov:deformationbook}.
\end{proof}
 The first iteration steps are
  \begin{align}\label{curvIterations}
    r^{(3)} &= \delta^{-1}\hat R = -\frac 1{16}R_{ijkl}
      y^iy^j(y^kdx^l - y^ldx^k) \nonumber\\
     &= -\frac 18 R_{ijkl}y^iy^jy^k dx^l\\
    \nabla \delta^{-1}\hat R &= -\frac 18\nabla_m R_{ijkl} y^iy^jy^k
      dx^m\wedge dx^l \nonumber\\
    \delta^{-1}\nabla r^{(3)}&= -\frac 1{40} y^n\nabla_m R_{ijkl} y^iy^jy^k
       \big(\delta_{m,n}dx^l - \delta_{n,l}dx^m\big) \nonumber\\
     &= \frac 1{40}\Big(\nabla_m R_{ijkl} - \nabla_l R_{ijkm} +
       R_{ijkn} \Gamma^n_{ml}\Big) y^iy^jy^ky^l dx^m \nonumber\\
    r^{(4)} &= -\frac 1{40}\tilde\nabla_m R_{ijkl}y^iy^jy^ky^mdx^l
    +\mathcal O(\hbar^{5/2}),\nonumber
  \end{align}
 where $\tilde\nabla$ is the product connection on $\mathcal
 W\otimes \Lambda$, thus acts the same way on $y^i$ and $dx^i$. In
 the next to last line, $\nabla$ also acts on the $y^i$ outside the brackets, and
 the term $\nabla_m R_{ijkl}y^iy^jy^ky^ldx^m$ vanishes due to the antisymmetry of $R_{ijkl}$ under exchange of
 $k$ and $l$.\\
 We have thus constructed a (formal) flat connection on $\mathcal
 W$, and want to identify the quantum operators with the set of flat sections of
 $\mathcal W$ with respect to $D$, i.e. those satisfying $D\hat
 f=0$. As $D=\nabla- \delta+\frac i\hbar [r,\cdot]$, this equation can be written
 in the form
  \begin{equation}\label{flacheObservBedg}
        \delta \hat f=\nabla \hat f + \frac i\hbar [r,\hat f].
  \end{equation}
 We denote the set of flat sections by $\Gamma_D(\mathcal W)$.
\begin{theo}[Fedosov]
 To every $f\in C^\infty(M)[[\hbar]]$ there is exactly one $\hat
 f\in \Gamma_D(\mathcal W)$ such that $\hat f_0=f$.
\end{theo}
\begin{proof}[Sketch of proof]
 For a 0-form $\hat f$ we have $\delta^{-1}\hat f=0$ and $\hat f_{00}=\hat f_0$,
 so that the decomposition (\ref{decompFed}) becomes $\hat f=\hat f_0 + \delta^{-1}\delta \hat
 f$. Then equation (\ref{flacheObservBedg}) implies
  \begin{equation}
   \hat f=\hat f_0+\delta^{-1}\big(\nabla \hat f+\frac i\hbar[r,\hat
   f] \big).
  \end{equation}
 Again this equation has a unique solution, which can be determined by
 iteration:
  \begin{equation}\label{OperatorIterationGlg}
   \hat f^{(0)}=\hat f_0=f,\quad \hat f^{(n+1)}=\hat
   f_{0}+\delta^{-1} \big(\nabla \hat f^{(n)}+\frac i\hbar [r,\hat
   f^{(n)}]\big).
  \end{equation}
 For the proof that $\hat f$ indeed solves (\ref{flacheObservBedg}) we
 again refer to \cite{fedosov:deformationart,fedosov:deformationbook}.
\end{proof}

 The first iterations give (we always determine $\hat f^{(n)}$
 up to order $\hbar^{n/2}$ only, because higher order terms are not
 stable under iteration):
  \begin{align}\label{OperatorIteration}
    \hat f^{(1)} + \mathcal O(\hbar^1)&=f+\delta^{-1}\nabla f = f+y^k\nabla_k f = f+\partial_k fy^k\nonumber \\
    \delta^{-1} \nabla \partial_k fy^k&= \delta^{-1}
    \big(dx^j\nabla_jy^k \nabla_k f\big)=\frac 12 y^j\nabla_j
      y^k\nabla_k f \nonumber \\
     \hat f^{(2)} &= f+\partial_kfy^k + \frac 12\nabla_j\partial_k fy^jy^k
      +\mathcal O(\hbar^{3/2})\nonumber\\
    [r^{(0)},\hat f^{(1)}] &= -\frac 18 R_{abcd}\partial_k f[y^ay^by^c,y^k]
       dx^d\\
     &= -\frac {i\hbar}8 R_{abcd} \partial_k f\big( y^ay^b\omega^{ck} +
        y^ay^c\omega^{bk} + y^by^c\omega^{ak} \big)dx^d \nonumber\\
     \frac i\hbar \delta^{-1}[r^{(0)},f^{(1)}] &= \frac
     1{24}R_{abcd}\omega^{ck}\partial_k f y^ay^by^d \nonumber\\
    \hat f^{(3)} &= f+y^k\nabla_kf + \frac 12y^j\nabla_jy^k\nabla_k
       f + \frac 16 y^i\nabla_iy^j\nabla_j y^k\nabla_k f \nonumber\\
       &\qquad\qquad +\frac 1{24} R_{abcd}\omega^{ck}\partial_k f y^ay^by^d + \mathcal O(\hbar^2).\nonumber
  \end{align}
 An explicit calculation of the covariant derivatives shows that in
 third order one has
 \begin{align}\label{OperatorIteration2}
    \hat f^{(3)}&=f+\partial_kfy^k + \frac
       12\omega_{jl} (\nabla_k X_f)^l y^jy^k + \nonumber\\
      &\qquad+ \frac 16\big[\omega_{jl}(\nabla_i\nabla_kX_f)^l +
      \Gamma_{ijl}(\nabla_kX_f)^l-\frac
      14R_{ijkl}X_f^l\big]y^iy^jy^k,
 \end{align}
 which only contains covariant derivatives of the vector field
 $X_f$, thus the $\nabla$-operators do not act on the $y^i$. The
 quantization map
  $$ (\pi_\otimes^0)^{-1} : C^\infty(M)[[\hbar]] \rightarrow
  \Gamma_D(\mathcal W),\ f\mapsto \hat f$$
 allows for the definition of a star product on
 $C^\infty(M)[[\hbar]]$, namely
  \begin{equation}\label{StarProdDefi}
         f\ast g:= \pi_\otimes^0(\hat f\circ \hat g).
  \end{equation}
 Using the explicit expression (\ref{OperatorIteration2}), as well
 as the relation
  \begin{equation}\label{DefoQuantContraction}
        \pi_\otimes^0(y^{a_1} \dots y^{a_k} \circ y^{i_1}\dots y^{i_k}) =
  \Big( \frac {i\hbar}2\Big)^k \sum_{\pi \in S_k} \omega^{a_1
  i_{\pi(1)}} \dots \omega^{a_k i_{\pi(k)}},
  \end{equation}
 we arrive at
  \begin{equation}\label{StarProdExpl}
          f\ast g = fg - \frac {i\hbar}2 \omega(X_f,X_g) + \frac{\hbar^2}4
      (\nabla_jX_f)^b(\nabla_bX_g)^j+ O(\hbar^3).
  \end{equation}
 The conditions for a star product mentioned at the beginning of the
 section are easily checked, using the fact that the Poisson bracket can be expressed as
 $\{f,g\}= \omega(X_g,X_f)$. For $M=\mathbb R^{2n}$ with $\Gamma=0$ one obtains
  \begin{equation}\label{vollerOperatorflach}
    \hat f = \sum_{k=0}^\infty \frac 1{k!}(\partial_{i_1} \dots
    \partial_{i_k} f) y^{i_1}\dots y^{i_k},
  \end{equation}
 and thus the Groenewold-Moyal product \cite{Groenewold_QM,
 moyal:original}
  \begin{equation}
   f\ast g(x)=\exp\Big(\frac {i\hbar}2\omega^{ij}\frac \partial{\partial
   x^i}\frac \partial{\partial y^j}\Big) f(x)g(y)\bigg|_{y=x}.
  \end{equation}

 In general the result seems to depend on the symplectic connection
 chosen, but different connections lead to equivalent star products,
 where two algebras $\mathcal A_1,\mathcal A_2$ are considered
 equivalent if there is a formal sum
  $$ S=1+ \sum_{k\geq 1} \hbar ^k S_k$$
 of differential operators $S_k :\mathcal A_1\rightarrow \mathcal
 A_2$, with $Sf\ast_2 Sg=S(f\ast_1 g)$. Still it is possible
 to obtain nonequivalent star products by allowing for higher order
 terms in the curvature $\Omega=\frac i\hbar[\tilde \Omega,\cdot]$:
  $$\tilde \Omega = \omega  + \sum_{k\geq 1} \hbar^k \omega_k,$$
 with closed 2-forms $\omega_k$. The equivalence class of the
 resulting star product depends on the cohomology class of the
 $\omega_k$, and as every star product can be constructed this way,
 $H^2(M,\mathbb C)[[\hbar]]$ parametrizes inequivalent star products
 on $(M,\omega)$. The class $[\omega] + \sum_k\hbar^k [\omega_k]$ is
 called Fedosov class.

\section{Representations of the deformed algebra}
\begin{bsp}[The symplectic Clifford algebra]
 We consider a vector space case $M=V$ with constant
 symplectic form $\omega$, and $\Gamma=0$. The following
 construction is in complete analogy to the construction of
 representations of the Clifford algebra over a metric vector space \cite{Varadarajan_SUSY}.
 The Weyl algebra is $W(V^\ast)=$Sym$(V^\ast_\mathbb C)[[\hbar]]$ with commutation relation
 $[v,w]=i\hbar\omega^{-1}(v,w)$. A given polarization $P\subset V_\mathbb
 C$ is transferred to the dual space $V^\ast _\mathbb C$ by means of
 the musical isomorphism
  $$ \flat : V \rightarrow V^\ast , \
  v^\flat(w)=\omega(w,v)=\omega_{ab}w^a v^b.$$
 For a real polarization we consider $P$ as a subset of $V$, and
 choose a Darboux basis $q^i,p_j$ of $V^\ast$ such that
 $P^\flat=$ span$\{q^1,\dots,q^n$\}. A
 representation of $W(V^\ast)$ is defined on Sym$(P^\flat_\mathbb
 C)[[\hbar]]$, i.e. on complex polynomials in the $q^i$, by
 \begin{align}
    \sigma(q)\psi&=q\psi \\
    \sigma(p)\psi&=[p,\psi]=  i\hbar \iota(p)\psi,\nonumber
 \end{align}
 for $q\in P_\mathbb C^\flat$, $p\in$ span$\{p_1,\dots,p_n\}$, and
 $$ \iota(p)q^{i_1}\dots q^{i_k}=\sum_{j=1}^k
 \omega^{-1}(p,q^{i_j})q^{i_1}\dots \breve q^{i_j}\dots q^{i_k}.$$
 For a basis element $p_i$ we obtain for $\sigma(p_i)$ the
 well-known Schrödinger operator $\frac \hbar i\partial_{q^i}$.
 Allowing in a slight generalization for arbitrary functions $\psi$ in
 $q$, we get the Schrödinger (or position space) representation of
 quantum mechanics.\\
 Consider now a totally complex polarization $P\subset V_\mathbb
 C$, such that $D=(P\cap \overline P)\cap V=0$. We choose a
 complex basis $\{z^i,\overline z^i\}$ of $V_\mathbb C^\ast$ such that $P^\flat=$ span$\{z^1,\dots,z^n\}$. Now the
 representation space is Sym($\overline P^\flat)[[\hbar]]$, i.e.
 polynomials in $z$, or holomorphic functions, and
  \begin{align}
    \rho(z)\psi &= z\psi ,\\
    \rho(\overline z)\psi &=[\overline z,\psi]= i\hbar \iota(\overline z)\psi, \nonumber
  \end{align}
 for $z\in \overline P^\flat$, $\overline z\in  P^\flat$, and
  $$ \iota(\overline z)z^{i_1}\dots z^{i_k}=\sum_{j=1}^k
  \omega^{-1}(\overline z,z^{i_j})z^{i_1} \dots \breve z^{i_j}\dots
  z^{i_k}.$$
 This way we get the Fock representation of $W(V^\ast)$, with
 $\rho(\overline z^i)=2\hbar \partial_{z^i}$, in case $\omega^{-1}$
 assumes the normal form $\omega^{-1}(z^i,\overline
 z^j)=2i\delta^{ij}$. The case of a general
 polarization $P\subset V_\mathbb C$ is treated analogously, with
 representation space Sym($\overline P^\flat)[[\hbar]]$, which can
 also be thought of as a space of complex functions on $V/D$, where $D=(P\cap
 \overline P)\cap V$.
\end{bsp}
 The example suggests that in order to obtain classical quantum
 mechanics on a symplectic vector space one needs not only the star
 product, but also a polarization (though all polarizations are
 equivalent), which usually occurs in the framework of geometric
 quantization. It therefore strongly hints at the necessity to
 consider 'polarized deformation quantization' (a term apparently
 introduced by Bressler and Donin, \cite{bressler_PolDefQua}) also in the general
 case of a symplectic manifold, and thus to unify geometric and
 deformation quantization.

\bigskip

 It is our aim to define a representation of the deformed algebra
 $\mathcal A=(C^\infty(M)[[\hbar]],\ast)$ (\ref{StarProdDefi}) on the Hilbert space
 $\mathcal H_P$ (\ref{geomHilbertSpace}) of geometric quantization. We already remarked that
 it is not obvious whether this is possible, and if so, whether the
 undetermined structures (symplectic connection and polarization)
 have to satisfy some compatibility condition. Therefore, we
 restrict our attention to the special cases of a cotangent bundle
 and a Kähler manifold, where we have a natural choice for these
 structures. For a few observables the representation on $\mathcal
 H_P$ has already been defined by geometric quantization in (\ref{geomOperatorDefi}) as
  \begin{equation}
    \rho(f) =-i\hbar \big(\nabla_{X_f}\otimes \mathbf 1+\mathbf 1
    \otimes \mathcal L_{X_f}\big) + f\mathbf 1\otimes \mathbf 1,
  \end{equation}
 (which was denoted by $\tilde f$ before) in case the flow of $X_f$ preserves the polarization.
 In order to extend $\rho$ to an algebra homomorphism we must make
 sure that
  \begin{equation}\label{KompiBdg}
        \rho(f)\rho(g)=\rho(f\ast g)
  \end{equation}
 is satisfied whenever the
 flows of the Hamiltonian vector fields of $f,g$ and $f\ast g $
 preserve the polarization. Then we can define $\rho(f\ast g)$ by
 $\rho(f)\rho(g)$ if $f$ and $g$ satisfy the quantization condition
 but $f\ast g$ does not, and iterate this procedure to obtain the
 operators for a large class of functions.

\paragraph{Cotangent bundle} Let $(Q,g)$ be an $n$-dimensional, orientable semi-Riemannian
 manifold, and $M=T^\ast Q$. $M$ carries a natural symplectic form,
 given by $\omega=dp_i\wedge dq^i$, where $\{q^i\}$ are coordinates on
 $Q$, and $\{q^i,p_j= dq^j\}$ the induced coordinates on $T^\ast Q$.
 The vertical polarization of $(M,\omega)$ is spanned by the vector
 fields $\partial_{p_i}$, so that polarized functions on $M$ can be
 identified as functions on $Q$. For the
 symplectic potential given by $\theta=p_idq^i$, wave
 functions are just polarized functions. If $f$ and $g$ are
 polarized, then so is $fg$, and we have $\rho(f)\rho(g)=\rho(fg)$,
 where the operators are given by multiplication with the
 corresponding functions. In this case the compatibility condition
 (\ref{KompiBdg}) becomes
  \begin{equation}
    f\ast g=fg.
  \end{equation}
 Further we consider the case that $f$ is linear in $p$,
 $f(q,p)=a^i(q)p_i$, and $g$ polarized. Then we have $\rho(f)= $
 Sym($a^i(\tilde q)\tilde p_i)$, which we can be reordered as
  $$ \rho(f)=\tilde p^ia _i(\tilde q) +
  \frac{i\hbar}2(\partial_ia^i)(\tilde q).$$
 This gives
 \begin{align}
   \rho(f)\rho(g)&=\Big[\tilde p_ia^i(\tilde q)+ \frac{i\hbar}2(\partial_ia^i)(\tilde
     q)\Big]g(\tilde q)\nonumber \\
   &= \tilde p_ia^i(\tilde q)g(\tilde q)+ \frac{i\hbar}2\big(\partial_i(a^ig)\big)(\tilde
     q) -  \frac{i\hbar}2 a^i(\tilde q)\partial_i g(\tilde q)\\
   &=\rho\Big(fg-
   \frac{i\hbar}2\partial_{p_k}f\partial_{q^k}g\Big),\nonumber
 \end{align}
 and condition (\ref{KompiBdg}) becomes $ f\ast g=  fg +\frac{i\hbar }2\{f,g\}$.
 Analogously one obtains $g\ast f= gf +\frac{i\hbar }2\{g,f\}$. In
 case $fg$ is a polynomial (or power series) of degree $\geq2$ in $p$, then also $f\ast
 g=fg+\mathcal O(\hbar^1)$ is a polynomial (or power series) of degree at least 2 in
 $p$, and it follows that $f\ast g$ is not directly quantizable.
 Thus the compatibility condition only applies to the
 cases considered so far, and we have
\begin{theo}\label{KompiTheorem}
 A star product (defined by a symplectic connection) on $(T^\ast Q,\omega)$ is compatible
 with geometric quantization of the vertical polarization in the sense
 that (\ref{KompiBdg}) holds, iff for any polarized functions $f,g$ and
 any function $h$ such that $X_h$ preserves the polarization, the
 following are satisfied
  \begin{enumerate}
    \item $f\ast g=fg$,
    \item $f\ast h= fh+\frac {i\hbar}2\{f,h\},\qquad h\ast f=hf +
    \frac{i\hbar}2 \{h,f\}$.
  \end{enumerate}
\end{theo}
 We recall that the first terms of the star product (\ref{StarProdExpl}) are $f\ast
 g=fg + \frac{i\hbar}2\{f,g\}+ \mathcal O(\hbar^2)$, so that the
 conditions are equivalent to the vanishing of all higher order
 terms. Bordemann et. al. \cite{BordemannNW-StarprodKot} have shown that there is a natural
 lift of any torsion-free connection on $Q$ to $T^\ast Q$, such that
 the resulting connection is torsion-free, symplectic, and leads to
 a homogeneous star product, i.e. for two homogeneous polynomials of degree $k$
 and $l$ in $p$, their star product is a homogeneous polynomial of
 degree $k+l$ in $p$ and $\hbar$. More precisely, $\mathcal H=p_i\frac \partial
 {\partial p_i} +\hbar \frac{\partial}{\partial\hbar}$ satisfies
  $$ \mathcal H(f\ast g) = (\mathcal H f)\ast g+ f\ast (\mathcal Hg). $$
 In particular, this implies that the
 conditions of theorem \ref{KompiTheorem} are satisfied.
 Further, any function polynomial in the momenta
 can be expressed as a finite sum of star products of functions that
 are affine-linear in $p$. Explicitly the lift is given by
  \begin{align}\label{kotbdlIndZshg}
    \Gamma^{k}_{ij}&=-\Gamma^{\overline j}_{i\overline
       k}=-\Gamma^{\overline i}_{\overline kj}= \tilde\Gamma^k_{ij} ,\nonumber\\
    \Gamma^{\overline k}_{ij}=\frac {p^a}3& \Big(2\tilde
       \Gamma^a_{jl}\tilde\Gamma^l_{ki} - \partial_j\tilde
       \Gamma^a_{ki} + \text{cycl.}(ijk)\Big), \\
    R^l_{kij}=-R^{\overline k}_{\overline lij}&=\tilde
      R^l_{kij},\qquad R^{\overline l}_{k i \overline j} =\frac
      13\big(\tilde R^j_{lki}+\tilde R^j_{kli}\big), \nonumber\\
    R^{\overline i}_{jkl} = \frac {p^a}3 \Big( \nabla_i \tilde
      R^a_{jlk} - &3\tilde \Gamma^a_{im}\tilde R^m_{jlk} - \tilde
      \Gamma^a_{lm}\tilde R^m_{ijk} + \tilde\Gamma^a_{km}\tilde
      R^m_{ijl} + (i\leftrightarrow j)\Big),\nonumber
  \end{align}
 where $\tilde \Gamma^k_{ij}$ and $\tilde R^i_{jkl}$ are the
 Christoffel symbols and curvature tensor on $Q$, and $\Gamma^k_{ij},
 R^i_{jkl}$ the lifted objects on $T^\ast Q$. The indices run from
 1 to $n$, an unbarred index $i$ stands for $q^i$, a barred index
 $\overline i$ for $p_i$. Components not
 listed vanish. This lift has the
 property that every geodesic on $T^\ast Q$ is projected to a
 geodesic on $Q$ \cite{plebanski:zushgLiftKotanBdl}. As $Q$ is a
 semi-Riemannian manifold there is a natural connection on $Q$ to start
 with, the Levi-Civita connection of $g$.\\
 Thus condition (\ref{KompiBdg}) is satisfied for the most natural
 choice of star product on $T^\ast Q$, and every function polynomial
 in the momenta can be quantized, although the explicit calculation
 of the corresponding operator becomes tedious for high order
 polynomials. It should also be emphasized that the quantization
 really depends on the metric $g$ on $Q$ (or the choice of a
 torsion-free connection), and not just on $\omega$.

\begin{bsp}[Quadratic observables]
 We determine the kinetic energy operator, i.e. the quantization
 of $  g^{ab}(q)p_ap_b ,$
 which we split into two parts as $g^{ab}p_ap_b=f_ah^a$, where
  \begin{equation}
   f_a(q,p)=p_a,\qquad h^a(q,p)=g^{ab}p_b.
  \end{equation}
 Due to the homogeneity of the star product we have
  $$ h^a\ast f_a=h^af_a + \frac{i\hbar}2\{h^a,f_a\} + \frac{\hbar^2}4
   (\nabla_\nu X_{h^a})^\mu(\nabla_\mu X_{f_a})^\nu , $$
 which, in the representation $\rho$, allows us to solve for
 $\rho(h^af_a)$:
  \begin{equation}\label{quadOpGlg1}
   \rho\big(g^{ab}p_ap_b\big)= \rho(h^a)\rho(f_a)-  \frac{i\hbar}2\rho\big(\{h^a,f_a\}\big) - \frac{\hbar^2}4\rho\Big(
  (\nabla_\nu X_{h^a})^\mu(\nabla_\mu X_{f_a})^\nu \Big).
  \end{equation}
 A simple calculation yields
 \begin{align*}
  \{h^a,f_a\} &= p_b\partial_{q^a}g^{ab}, \\
   (\nabla_\nu X_{h^a})^\mu(\nabla_\mu X_{f_a})^\nu &= 2g^{ab}\tilde
  \Gamma^i_{ja}\tilde \Gamma^j_{ib},
 \end{align*}
 where (\ref{kotbdlIndZshg}) has been used for the second equation.
 Now the operators on the right hand side of (\ref{quadOpGlg1}) are
 determined by geometric quantization:
  \begin{align*}
   \rho(f_a) = \tilde p_a,&\qquad \rho(h^a)= g^{ab}(\tilde
     q)  \tilde p_b- \frac{i\hbar}2 (\partial_b g^{ab})(\tilde q),\\
   \rho(\{h^a,f_a\}) &= (\partial_a g^{ab})(\tilde q)  \tilde p_b- \frac
    {i\hbar}2 (\partial_a\partial_b g^{ab})(\tilde q), \\
   \rho\Big((\nabla_\nu X_{h^a})^\mu&(\nabla_\mu X_{f_a})^\nu\Big)= 2 g^{ab}\tilde
  \Gamma^i_{ja}\tilde \Gamma^j_{ib} (\tilde q).
  \end{align*}
 Inserting these expressions into (\ref{quadOpGlg1}) gives
 \begin{align}
   \rho\big(g^{ab}p_ap_b\big) &= g^{ab}(\tilde
     q)  \tilde p_b\tilde p_a - \frac{i\hbar}2\big( (\partial_bg^{ab})(\tilde q)\tilde p_a +
     (\partial_ag^{ab})(\tilde q)\tilde p_b \big) \\
   &\quad - \frac{\hbar^2}4 \Big[(\partial_a\partial_bg^{ab})(\tilde q) +
     2g^{ab}\tilde \Gamma^i_{ja}\tilde \Gamma^j_{ib} (\tilde
     q)\Big].\nonumber
 \end{align}
 Here we make use of the explicit form of $\tilde q^a$ and $\tilde
 p_b$ (\ref{kanOpsGeomQuant}):
  $$ \tilde q^a = \hat q^a,\qquad \tilde p_a =\hat p_a -
  \frac{i\hbar}4g^{bc}\partial_ag_{bc}(\hat q),$$
 where $\hat q^a$ is multiplication by $q^a$ and $\hat
 p_b=-i\hbar\partial_{q^b}$. $\tilde p_a$ can also be written as
 $\tilde p_a = \hat p_a - \frac {i\hbar}2g^{-1/2}\partial_a
 g^{1/2}$, with $g:=|\det g|$. Then the first term becomes
  $$g^{ab}\tilde p_a\tilde p_b= g^{ab}\hat p_a\hat p_b-i\hbar
  g^{ab}g^{-1/2}\partial_ag^{1/2}\hat p_b -
  \hbar^2g^{ab}g^{-1/2}\partial_a\partial_b g^{1/2},$$
 and we obtain
  \begin{align}
   \rho\big(g^{ab}p_ap_b\big) &= g^{ab}\hat p_a\hat p_b - i\hbar
   g^{-1/2}\partial_a(g^{1/2}g^{ab})\hat p_b - \\
    &\quad-\frac {\hbar^2}4\Big[\partial_a\partial_b
    g^{ab}+2g^{-1/2} \partial_a(g^{ab}\partial_b g^{1/2})+2g^{ab} \tilde \Gamma^i_{ja}\tilde
    \Gamma^j_{ib}\Big].\nonumber
  \end{align}
 The first two terms yield $-\hbar^2$ times the Laplace-Beltrami
 operator $\Delta=g^{-1/2}\partial_ag^{1/2}g^{ab}\partial_b$.
 In order to simplify the last four terms we use normal coordinates, which satisfy
 $\partial_a g_{ij}=\partial_ag^{ij}=\Gamma^k_{ij}=0$ in a fixed
 point, and are left with $-\frac {\hbar^2}4$ times
  $$ \partial_a\partial_b g^{ab} +2
  g^{-1/2}g^{ab}\partial_a\partial_b g^{1/2}.$$
 Using $\partial_a g^{1/2}=\frac 12g^{1/2}g^{kl}\partial_a g_{kl}$
 and the relation $\partial_a g^{kl}=-g^{km}g^{ln}\partial_ag_{mn}$,
 which directly follows from $g^{kl}g_{lm}={\delta^k}_m$, we arrive
 at
  $$  \rho\big(g^{ab}p_ap_b\big) = -\hbar ^2\Delta
  -\frac{\hbar^2}4g^{ac}g^{bd}[\partial_a\partial_b
  g_{cd}-\partial_a\partial_cg_{bd}].$$
 Here we recognize the scalar curvature of $Q$ (in normal
 coordinates):
   \begin{align}
 \tilde R&= g^{ik}g^{jl}\tilde R_{ijkl} \nonumber\\
  &=\frac 12 g^{ik}g^{jl} \big(\partial_l\partial_i g_{jk} +\partial_j\partial_k
   g_{il} -\partial_i\partial_k g_{jl}- \partial_l\partial_j
   g_{ik}\big) \\
  &= g^{ik}g^{jl} \big(\partial_l\partial_i g_{jk}  -\partial_i\partial_k
  g_{jl}\big).\nonumber
 \end{align}
 Thus the final result reads
  \begin{equation}
    \rho\big(g^{ab}p_ap_b\big)=-\hbar^2\Big(\Delta-\frac{\tilde
    R}4\Big).
  \end{equation}
 It is of the well-known form $\Delta-\alpha \tilde R$; there has
 been a long debate over the correct value of $\alpha$ though. E.g. our result
 $\alpha=1/4$ was obtained by a variational principle in
 \cite{kamokawai_schwingervariation}, whereas Cheng got $\alpha=1/3$
 from a path integral formalism \cite{Cheng_sglausPfadint}, and Woodhouse gives an argument for
 $\alpha=1/6$ in the framework of geometric quantization \cite{woodhouse:quantisierung}. In the
 formalism advocated here, we see that some inevitable arbitrariness
 lies in the choice of the half-form (here $\sqrt{g^{1/2}d^nq}$), and that
 different symplectic connections would lead to different
 results. Still, our choices are the most natural ones.
\end{bsp}

\paragraph{Kähler manifolds}
 For the choice $\theta=-i\frac{\partial K}{\partial z^k}dz^k $ and $\mu=d^n
 z$ wave functions are polarized, the situation is thus very similar
 to the cotangent bundle case. Functions of the form $u^k(z)\frac{\partial K}{\partial
 z^k} + v(z)$ can be quantized directly in geometric quantization,
 and
  $$ \tilde z^a\psi(z)=z^a\psi(z),\qquad
  (\widetilde{\partial_{z^a}K})\psi(z)=\hbar\partial_{z^a} \psi(z).$$
 Theorem (\ref{KompiTheorem}) continues to hold if we replace
 $T^\ast Q$ by a Kähler manifold $M$ and the vertical polarization
 by the holomorphic Kähler polarization. Thus compatibility is
 fulfilled if for $f,g$ holomorphic also $f\ast g$ is
 holomorphic, and the star product of a function affine-linear in $\frac{\partial K}{\partial
 z^k}$ with a holomorphic function is affine-linear again.
 There is a natural symplectic connection on Kähler manifolds, given
 by the Levi-Civita connection
 of the Kähler metric $g=\frac 12
 \frac{\partial^2K}{\partial z^j\partial \overline
 z^k}\big(dz^j\otimes d\overline z^k+ d\overline z^k\otimes
 dz^j\big)$. Defining
  \begin{equation}
    A_{j\overline k}=\frac{\partial^2K}{\partial z^j\partial \overline
 z^k}, \qquad  A_{j\overline k}A^{\overline kl} = A_{k\overline j}A^{\overline lk} = {\delta_j}^l,
  \end{equation}
 the symplectic form becomes $\omega_{j\overline k}=-\omega_{\overline
 kj}=iA_{j\overline k}$, and its inverse $\omega^{j\overline
 k}=-\omega^{\overline kj }=iA^{\overline kj}$.
  The only non-vanishing connection coefficients are given by \cite{jost:RiemGeom}
  \begin{equation}\label{Kaehlerconnection}
         \Gamma^k_{ij} = A^{\overline lk}\partial_{j}A_{i\overline
  l},\qquad \Gamma^{\overline k}_{\overline {ij}}=A^{\overline
  kl}\partial_{\overline j}A_{m\overline i},
  \end{equation}
 where $\partial_j=\partial_{z^j}$ and $\partial_{\overline
 j}=\partial_{\overline z^j}$,
 and the curvature coefficients are determined by
  \begin{equation}\label{KaehlerKrmg}
               R_{k\overline li\overline j}= i\partial_i\partial_{\overline l}
      A_{k\overline j} - iA^{\overline n m}\partial_i A_{k\overline
      n} \partial_{\overline l}A_{m\overline j}
  \end{equation}
 plus the symmetries
  \begin{align}
    R_{k\overline li\overline j}=-R_{k\overline l\overline ji}&=- R_{\overline j i\overline l
    k}= R_{\overline j i k\overline l}.
  \end{align}
 Other components vanish. Here
 $R_{\mu\nu\kappa\lambda}=\omega_{\mu\tau}R^\tau_{\nu\kappa\lambda}$,
 whereas usually in the context of Kähler manifolds the index is
 lowered with the metric.
 Due to the obvious relation $\partial_iA_{j\overline
 k}=\partial_jA_{i\overline k}$ and (\ref{KaehlerKrmg}), or the first Bianchi identity, also
 the following symmetries hold:
 \begin{equation}
     R_{\overline kl\overline ij}=R_{\overline kj\overline il},\qquad
    R_{k\overline li\overline j} =R_{k\overline ji\overline l}.
 \end{equation}
 This allows us to determine $r^{(3)}$ (\ref{curvIterations}):
 \begin{align}
   r^{(3)} &= -\frac 18 R_{\kappa\lambda\mu\nu}y^\kappa y^\lambda y^\mu
       dx^\nu\nonumber\\
    &= -\frac 18\Big[ R_{k\overline li\overline j}\hat z^k\hat{\overline
      z}^l\hat z^i d\overline z^j + R_{k\overline l\overline ji}\hat z^k\hat {\overline
      z}^l\hat{\overline z}^j dz^i + R_{\overline lki\overline j}\hat z^k\hat {\overline
      z}^l\hat z^i d\overline z^j + R_{\overline lk\overline ji}\hat z^k\hat {\overline
      z}^l\hat{\overline z}^j dz^i \Big]\\
    &=-\frac 14R_{k\overline li\overline j}\hat z^k\hat {\overline
    z}^l\big(\hat z^id\overline z^j-\hat{\overline z}^jdz^i\big),\nonumber
 \end{align}
 where $\hat z^k,\hat{\overline z}^l$ are the local generators of $T^\ast M$,
 collectively denoted $y^\mu$ before, and satisfying the canonical commutation
 relation $[\hat{\overline z}^l,\hat z^k]_\circ= i\hbar\omega^{\overline lk}=\hbar
 A^{\overline lk}$. As before for $y^\mu$, here $\hat z^j\hat{\overline z}^k$ denotes
 the symmetrized (!) product of $\hat z^j$ and $\hat{\overline
 z}^k$, whereas simple composition of operators is $\hat z^j\circ \hat{\overline
 z}^k$. We calculate the first terms of the local operators
 $\hat f,\hat h \in \Gamma_D(\mathcal W)$ corresponding to the
 functions $f(z,\overline z)=\overline w_a z^a$, $h(z,\overline z)=
 -i\frac{\partial K}{\partial z^m}$, according to
 (\ref{OperatorIterationGlg}), and (\ref{StarProdDefi}), in order to prove that there are
 no contributions to $f\ast h$ in order $\hbar^2$ and $\hbar^3$. We could as well
 use the solution (\ref{OperatorIteration}) to (\ref{OperatorIterationGlg}) directly, but performing the iteration
 explicitly allows us to drop a few terms that do not contribute to
 the star product in third order. Introducing the
 $\hbar$-homogeneous elements $\hat f_{(n)}:=\hat f^{(n)}-\hat
 f^{(n-1)}$ (mod $\hbar^{(n+1)/2})$, we have $\hat f_{(1)}=
 \overline w_a\hat z^a $, and
 \begin{align}
   \hat f_{(2)}&= -\frac 12\overline w_a
       \Gamma_{bc}^a\hat z^b\hat z^c \nonumber\\
   \hat f_{(3)}&\sim -\frac 16\overline w_a\partial_{\overline
     d}\Gamma_{bc}^a \hat z^b\hat z^c\hat{\overline z}^d + \frac
     i\hbar \delta^{-1}[r^{(3)},\hat f^{(1)}]\\
    &= -\frac 16\overline w_a\partial_{\overline d}\Gamma_{bc}^a \hat z^b\hat z^c\hat{\overline z}^d
     +\frac 1{12} \overline w_a R^a_{bc\overline d} \hat z^b\hat
     z^c\hat {\overline z}^d.\nonumber
 \end{align}
 $\sim$ means equality up to terms not containing any operator
 $\hat {\overline z}^i$, which will not contribute. Observing that
 $\partial_{\overline d}\Gamma^a_{bc} = R^a_{c\overline d b}=-R^a_{bc\overline d}$, we
 obtain
  \begin{equation}
    \hat f_{(3)}\sim \frac 1{4} \overline w_a R^a_{bc\overline
    d}\hat z^b\hat z^c\hat{\overline z}^d.
  \end{equation}
 Further $\hat h_{(1)}=-i \partial_m\partial_a K \hat z^a - i
 A_{m\overline b}\hat{\overline z}^b$, and
  $$ i\hat h_{(2)}\sim \frac 12\partial_a
  A_{m\overline b}\hat z^a\hat {\overline z}^b+\frac 12 \hat
  {\overline z}^c \underset 0{\underbrace{\nabla_{\overline c}(A_{m\overline b}\hat
  {\overline z}^b)}} + \frac 12\partial_a A_{m\overline b}\hat z^a\hat
  {\overline z}^b.$$
 The vanishing of $\nabla_{\overline c}(\omega_{m\overline b}\hat
 {\overline z}^b)$ is a consequence of the connection being
 symplectic: $\nabla \omega=0$, and the vanishing of some
 connection coefficients (\ref{Kaehlerconnection}), and is easily
 checked directly. We are left with $\hat h_{(2)}\sim-i \partial_a
 A_{m\overline b}\hat z^a\hat {\overline z}^b$. Now we define
 $\approx$ to be equality mod terms containing less than two
 operators $\hat{\overline z}^i$, and get
  \begin{align}
   \hat h_{(3)}&\approx -\frac i3 \hat {\overline
    z}^c\nabla_{\overline c}\big(\partial_a A_{m\overline b}\hat
    z^a\hat{\overline z}^b\big) + \frac i\hbar
    \delta^{-1}[r^{(3)},h^{(1)}]\nonumber\\
   \approx -\frac 13&R_{m\overline ba\overline c}\hat{\overline z}^b\hat{\overline
     z}^c \hat z^a + \frac 1{12}R_{m\overline ba\overline c}\hat z^a\hat{\overline
     z}^b \hat{\overline z}^c = -\frac 14 R_{m\overline ba\overline c}\hat z^a\hat{\overline
     z}^b \hat{\overline z}^c
  \end{align}
 From the definition of the star product (\ref{StarProdDefi}), relation
 (\ref{DefoQuantContraction}) and the results for $\hat f,\hat h$ we
 can read off that there is no contribution to
 $f\ast h$ in order $\hbar^2$, in agreement with the compatibility
 condition, and in order $\hbar^3$ we get the following result:
  \begin{equation}\label{dritteOrdnung1}
    \pi_\otimes^0\big(\hat f_{(3)}\circ \hat h_{(3)}\big)=- \frac{\hbar^3}{64}
    \overline w_a R^a_{bc\overline d}R_{m\overline nk\overline
    l}A^{\overline dk}A^{\overline nb}A^{\overline lc}.
  \end{equation}
 But there are further contributions in order $\hbar^3,$ from $\hat f_{(5)}\circ
 \hat h_{(1)}$ and $\hat f_{(1)}\circ \hat h_{(5)}$, due to the term
 $r^2$ in (\ref{abelscherZshgBedg}). $(r^{(n)})^2$ contains terms of
 the form $y^{i_1}\dots y^{i_n}\circ y^{j_1}\dots y^{j_n}$, which
 can be reordered, using $y^i\circ y^j=y^iy^j+ \frac
 {i\hbar}2\omega^{ij}$, such that only fully symmetric expressions
 $y^{i_1}\dots y^{i_k}$, $k=0,\ldots,2n$, occur. Obviously,
 only such terms in $\hat f_{(5)}$ contribute to $\pi_\otimes^0(\hat f_{(5)}\circ \hat
 h_{(1)})$ that contain exactly one operator $y^i$. The only term of
 this kind in $\hat f_{(5)}$ comes from $\delta^{-1}\frac
 i\hbar[r^{(5)},\hat f^{(1)}]$ in the iteration formula for $\hat f$
 (\ref{OperatorIterationGlg}), and there only the $\delta^{-1}\frac i\hbar\pi_\otimes ^0\big((r^{(3)})^2\big)$ part
 of $r^{(5)}$ contributes. Neglecting those terms not contributing, we have
  \begin{align*}
    (r^{(3)})^2 &=\frac {i\hbar^3}{32}R_{k\overline li\overline j}
  R_{a\overline b c\overline d} \omega^{\overline la} \omega^{\overline bk}\omega^{\overline
  jc} d\overline z^d\wedge dz^i, \\
   r_{(5)}= \delta^{-1}\frac i\hbar(r^{(3)})^2 &= -\frac {\hbar^2}{32}R_{k\overline li\overline j}
  R_{a\overline b c\overline d} \omega^{\overline la} \omega^{\overline bk}\omega^{\overline
  jc}(\hat{\overline z}^d dz^i - \hat z^id\overline z^d),
  \end{align*}
 which leads to
  \begin{align*}
    \delta^{-1}\frac i\hbar [r_{(5)},\hat f_{(1)}] &= -\frac {\hbar^2}{32}R_{k\overline li\overline j}
  R_{a\overline b c\overline d} \omega^{\overline la} \omega^{\overline bk}\omega^{\overline
  jc}\omega^{\overline dm}\partial_m f\hat z^i, \\
    \delta^{-1}\frac i\hbar [r_{(5)},\hat h_{(1)}] &= -\frac {\hbar^2}{32}R_{k\overline li\overline j}
  R_{a\overline b c\overline d} \omega^{\overline la} \omega^{\overline bk}\omega^{\overline
  jc}\omega^{\overline mi}\partial_{\overline m} h \hat{\overline z}^d,
  \end{align*}
 where we have also neglected a contribution containing $\partial_m
 h \omega^{\overline dm}\hat z^i$.
 According to the discussion above, the last two lines give the only components of $\hat
 f_{(5)}$ and $\hat h_{(5)}$ that contribute to $f\ast h$ in
 order $\hbar ^3$:
  \begin{align*}
   \pi_\otimes ^0\big(\hat f_{(5)}\circ \hat
   h_{(1)}\big) &= -\frac{i\hbar^3}{128} R_{k\overline li\overline j}
  R_{a\overline b c\overline d} \omega^{\overline la} \omega^{\overline bk}\omega^{\overline
  jc}\omega^{\overline dm} \partial_m f\partial_{\overline n}h
  \omega^{\overline n i} \\
   &= \pi_\otimes^0\big(\hat f_{(1)}\circ \hat h_{(5)}\big).
  \end{align*}
 Adding up, we get
 \begin{equation}
   \pi_\otimes^0\big(\hat f_{(5)}\circ \hat h_{(1)} + \hat f_{(1)}\circ
   \hat h_{(5)}\big) = \frac {\hbar^3}{64}\overline w_a
   R^a_{bc\overline d} R_{m\overline nk\overline l} A^{\overline nb} A^{\overline
   dm}A^{\overline lc},
 \end{equation}
 which exactly cancels the contribution of $\pi_\otimes ^0(\hat
 f_{(3)}\circ \hat h_{(3)})$ (\ref{dritteOrdnung1}),
 and we conclude
 \begin{equation}
   f\ast h=fh+ \frac {i\hbar}2\{f,h\} + \mathcal O(\hbar^4).
 \end{equation}
 For two holomorphic functions $f,g$ the operators $\hat f_{(1)}$
 and $\hat g_{(1)}$ both contain only unbarred operators, the same
 is true for $\hat f_{(2)},\hat g_{(2)}$, and the summands of $\hat f_{(3)}$
 contain at least two unbarred and at most one barred operator.
 Therefore the contraction $\pi_\otimes^0(\hat f^{(3)}\circ \hat
 g^{(3)})$ vanishes. Further, $\pi_\otimes^0(\hat f_{(5)}\circ\hat
 g_{(1)})$ vanishes, as $\pi_\otimes^1(\hat f_{(5)})$ contains only
 unbarred operators as well. Thus
  \begin{equation}
    f\ast g=fg+  \mathcal O(\hbar^4).
  \end{equation}
 The vanishing of the $\hbar^2$ component in $f\ast h$ can be
 clearly attributed to the choice of connection; it is a
 result of $\nabla_{\overline c}(A_{m\overline b}\hat {\overline
 z}^b)=\nabla_{a}\hat {\overline z}^b=\nabla_{\overline
 a}\hat z^b=0$. Compatibility in order $\hbar^3$ appears somewhat
 mysterious however, resulting from a cancelation of quite different contributions. \\
 We have thus established compatibility up to order $\hbar^3$ for a large class of functions, which,
 together with compatibility for cotangent bundles,
 provides strong evidence for exact compatibility. For a complete proof
 one could try to imitate the proof of Bordemann et. al. for
 cotangent bundles \cite{BordemannNW-StarprodKot}, in order to
 establish homogeneity also for Kähler manifolds. It appears, though,
 that the Kähler connection is not homogeneous in the sense of
 \cite{BordemannNW-StarprodKot}, so that at least some modifications are
 unavoidable. If the star product is compatible with the holomorphic
 polarization for general Kähler manifolds, then it must be
 compatible with the antiholomorphic polarization as well, for
 symmetry reasons. Further, if it is homogeneous in $\partial_a K$, then it
 must be homogeneous in $\partial_{\overline a}K$ too.  \\
 It can easily be seen from the iteration formula and the vanishing of connection
 coefficients with both barred and unbarred indices, that
 compatibility holds exactly if the curvature vanishes.

\section{Conclusion}
 We have proposed an interpretation of geometric and deformation
 quantization as complementary theories, in an attempt to overcome
 the difficulties arising in geometric quantization in quantizing
 the observables, and to define a natural representation of the
 deformed algebra. \\
 It turned out to work well if phase space is a
 cotangent bundle, as a consequence of homogeneity of the star product, which was established
 in \cite{BordemannNW-StarprodKot}.
 Still, even after restriction to the vertical
 polarization the quantization map is not uniquely determined. One further needs a
 metric or, alternatively, a torsion-free connection on
 configuration space. This explains the mentioned difficulties in
 geometric quantization quite naturally; although the framework for classical
 Hamiltonian mechanics is a symplectic manifold $(M,\omega)$
 \cite{Arnold_Mechanik,woodhouse:quantisierung},
 quantization is not determined by $\omega$ alone. One should
 keep in mind that in the typical situation of classical mechanics
 one has a configuration space $(Q,g)$ and as phase-space $T^\ast Q$, where the
 metric is needed to specify the kinetic term in the Lagrangian:
 $L_{kin}=g^{ij}p_ip_j$. Therefore, from a physical point of view
 the occurrence of the metric in the quantization process is not
 too surprising.\\
 Then it is tempting to assign a similar role to the metric and its associated Levi-Civita connection
 in the important case of general Kähler manifolds. But for these, we could only establish compatibility of the two
 theories up to a finite order in $\hbar$. The vanishing of the obstructions has partly been shown to be
 related to the special choice of connection.\\
 The case of general polarizations was not treated in this paper.

\bibliographystyle{Lit}
\bibliography{literatur}

\end{document}